\documentclass[11pt]{article}
\usepackage{amssymb}
\usepackage{amsmath}
\usepackage{amsthm}
\voffset = - 1.2cm
 \usepackage[colorlinks=true]{hyperref}

\hypersetup{urlcolor=blue, citecolor=red}

\newtheorem{theorem}{Theorem}[section]
\newtheorem{corollary}{Corollary}

\newtheorem{proposition}{Proposition}

\theoremstyle{definition}
\newtheorem{definition}[theorem]{Definition}

\begin{document}

\centerline{\bf \Large Superposition rules and second-order Riccati equations}
      
\medskip

\centerline{\scshape J.F. Cari\~nena$^*$ and J. de Lucas$^{**}$}
\medskip
 \centerline{$^*$Departamento de F\'isica Te\'orica and IUMA, Facultad de Ciencias, Universidad de Zaragoza,}
 \centerline{Pedro Cerbuna 12, 50.009, Zaragoza, Spain.}

\medskip
 \centerline{$^{**}$Institute of Mathematics, Polish Academy of Sciences,}
      \centerline{\'Sniadeckich 8, P.O. Box 21, 00-956, Warszawa, Poland.}
\medskip

\bigskip


\begin{abstract}
A {\it superposition rule} is a particular type of map that enables one to express the general solution of certain systems of first-order ordinary differential equations, the so-called {\it Lie systems}, out of generic families of particular solutions and a set of constants. The first aim of this work is to propose several generalisations of this notion to second-order differential equations. Next, several  results on the existence of such generalisations are 
given  and relations with the theories of Lie systems and quasi-Lie schemes are found. Finally, our methods are used  to study second-order Riccati equations and other second-order differential equations of mathematical and physical interest. 
\end{abstract}


{{\bf Primary:} {\it 34A26; Secondary: 34A34,53Z05.}}
 
{{\bf Keywords:} {\it Quasi-Lie scheme, Lie system, second-order system, second-order Riccati equation, superposition rule.}}

\section{Introduction}
The origin of the study of superposition rules can be traced back to 1893, when Lie, Guldberg, and Vessiot published a series of papers \cite{Gu93,LS,Ve93,Ve94} characterising and analysing systems of first-order ordinary differential equations admitting a superposition rule. This theory, broadly analysed during its beginning, was rarely treated for the next eighty years. Nevertheless, the interest in this topic has revived in the last few decades, and many works have recently been devoted to the analysis of its properties and to the study of its applications and generalisations \cite{BecGagHusWin90,CGL08,CGM07,CLR08,CLR08c,CarRam03,GGL08,Ib99,LW96,LP09,RSW97,Ve95,PW}. Among these works, we can 
emphasise the results described in \cite{CGM07,Ib99,PW}.

In spite of its interesting properties and applications in Mathematics, Control Theory, and Physics (see \cite{CGM07,CL08Non,CLR08,CLR08c,CarRam03,GGL08,OlmRodWin8687,RSW97}), the theory of Lie systems possesses a certain lack of applicability in many other problems described by differential equations that are not Lie systems. This motivates the generalisation of the methods for analysing Lie systems carried out here, and in previous works, so as to investigate non-Lie systems.

The theories of Lie systems \cite{CGM07,LS,PW} and quasi-Lie schemes \cite{CGL08,CLL08Emd} mainly concern the investigation of systems of first-order differential equations. Nevertheless, various attempts have been carried out to apply their methods to study second-order differential equations (SODEs) \cite{CL08Non,CLR08,CLR08c,GGL08,RSW97,Ve95}. These papers were rather more focused on applying the methods of the aforementioned theories to SODEs than on describing new theoretical results. This explains why, although these applications suggested the existence of a new type of (nonlinear) expressions describing general solutions of SODEs, no further discussion about such expressions and their properties was carried out.

In view of the above comments, the first aim of this work is to define some new notions describing, as particular cases, those expressions for the study of SODEs appearing   
in the very recent literature, namely, time-dependent and time-independent superposition rules for systems of SODEs. Furthermore, some theoretical results concerning the existence of such superposition rules are demonstrated, and several relations with the theories of Lie systems and quasi-Lie schemes are shown. In this way, our achievements allow us to clarify diverse procedures and techniques used in previous works to investigate SODEs.

After developing the theoretical part of our work, our results are applied to study families of SODEs appearing  broadly in the Physics and Mathematics literature \cite{Ar97,CGR09,CRS05,CLS05II,Da62,GGL08, In86,KL09,BLPS09}. More specifically, time-independent and time-dependent superposition rules are derived for such families using the theory of Lie systems and quasi-Lie schemes. This reduces the determination of the general solution of any instance of that families to the determination of four of its particular solutions. On one hand, this provides a modern fully geometric demonstration for a result sketched in a Vessiot's work \cite{Ve95} that was derived in {\it ad hoc} way. On the other hand, our work supplies a method to analyse the solutions of specific forms of time-dependent Li\'enard equations \cite{CLS05II,GGL08,BLPS09}, various interesting equations studied by Ince and Davis \cite{Da62,In86}, some modified Emden equations \cite{BLPS09}, certain second-order Riccati equations \cite{Ar97,CGR09,CRS05}, etc. Moreover, particular forms of the analysed SODEs arise in the analysis of fusion pellets \cite{AAE84}.

Among the aforementioned applications, particular attention is paid to the study of second-order Riccati equations. 
The interest is mainly due to two reasons. On one hand, these equations are widely studied in order to analyse their mathematical properties \cite{Ar97,CRS05,CGK08,Er77,FL66I,GL99,Da05}. On the other hand, second-order Riccati equations and some of their particular cases, e.g. certain modified Emden equations \cite{CLS05II,CLS05,BLPS09} or various Painlev\'e-Ince equations \cite{BFL91,CRS05,AAE84,Go50,KL09}, appear in many physical problems, like the problem of waves on shallow water. 

Another application of second-order Riccati equations deserves a special attention: these equations are an element of the 
so-called {\it Riccati chain} describing the B\"acklund transformations for various partial differential equations \cite{CGR09,GL99}. More specifically, for each one of such PDEs, a particular solution defines a one-parametric family of $n$-order Riccati equations. The solutions of the
  members of such a family allow one to build up a new family of solutions of the initial PDE \cite{GL99}. Obviously, 
  the time-dependent superposition rule derived in this work for second-order Riccati equations simplifies the study of the solutions
   of such a family and, in consequence, the determination of new solutions for those PDEs whose B\"acklund transformations are determined by these equations.

The plan of the paper is as follows. Section 2 mainly recalls the notions of the theory of Lie and quasi-Lie systems necessary to
follow the results of our paper. Additionally, the usefulness of quasi-Lie schemes to explicitly determine time-dependent superposition rules for systems of first-order differential equations is discussed. Section 3 is concerned with one of the novelties of our paper: the definition and analysis of the time-independent and the time-dependent superposition rule notions for systems of SODEs. In Section 4 our previous theoretical results are employed to derive a superposition rule for all the elements of a family of SODEs with many applications to Physics and Mathematics. A quasi-Lie scheme is used in Section 5  to
derive a time-dependent superposition rule for second-order Riccati equations. Finally, some calculations to prove various results 
of the paper are detailed in the Appendix.

\section{Lie systems and quasi-Lie schemes}\label{LSLS}

Here we mainly review some features about Lie systems, quasi-Lie schemes and time-dependent superposition rules for systems of first-order differential equations \cite{CGL08,CGL10,CGM07,CLL08Emd}. In addition, the use of quasi-Lie schemes for deriving time-dependent superposition rule is discussed. For simplicity, we restrict ourselves to analysing systems of differential equations on linear spaces and we skip various technical details that are not relevant to understand our procedures.

Lie initiated in \cite{LS} the nowadays called theory of Lie systems when started investigating the conditions that ensure that a system of first-order ordinary
differential equations of the form
\begin{equation}\label{FORD}
\frac{dx^i}{dt}=X^i(t,x),\qquad i=1,\ldots,n,
\end{equation}
admits a {\it superposition rule}, i.e. a $t$-independent function $\Phi:\mathbb{R}^{nm}\times\mathbb{R}^n\rightarrow
{\mathbb{R}}^n$, $x=\Phi(x_{1}, \ldots,x_{m};\lambda_1,\ldots,\lambda_n)$,
such that the general solution of the system can be written as
\begin{equation}\label{FirstSup}
x(t)=\Phi(x_{1}(t), \ldots,x_{m}(t);\lambda_1,\ldots,\lambda_n), 
\end{equation}
where
$\{x_{a}(t)\mid a=1,\ldots,m\}$ is any generic family of
particular solutions and $\lambda_1,\ldots,\lambda_n$ is a set of   $n$
constants. 

Note that the general solution of a linear homogeneous system of first-order differential equations in $\mathbb{R}$ cannot be cast into the form (\ref{FirstSup}) for every set of $n$ particular solutions: they must be linearly independent. In a similar way, for all known Lie systems (cf. \cite{LW96,PW}), their expressions (\ref{FirstSup}) only hold for certain families of particular solutions. More specifically, it is said that expression (\ref{FirstSup}) is valid for any `generic' family of $m$ particular solutions if there exists an open dense subset $U\subset\mathbb{R}^{nm}$ such that expression (\ref{FirstSup}) is satisfied for every set of particular solutions $x_{1}(t),\ldots,x_{m}(t)$ such that $(x_{1}(0),\ldots,x_{m}(0))$ lies in $U$. Obviously, `almost every' set of $m$ particular solutions holds the above property and this is, indeed, the approximate meaning of the term `generic' in the above definition. 
 
Lie found a characterisation describing systems of first-order differential equations admitting a superposition rule \cite{LS} and such a characterisation was recently reformulated in terms of the modern language of Differential Geometry \cite{CGM07}. In these terms, each system of the form (\ref{FORD}) is described by means of a time-dependent vector field $X(t,x)=\sum_{i=1}^nX^i(t,x)\partial/\partial x^i$ on $\mathbb{R}^n$. This description is used to establish that a first-order system (\ref{FORD}) admits a superposition rule if and only if its associated time-dependent
vector field $X(t,x)$ can be written as a linear combination
\begin{equation}
X(t,x)=\sum_{\alpha =1}^r b_\alpha(t)\, X_{\alpha}(x),
\label{Lievf}
\end{equation}
where the vector fields $$X_{\alpha}(x)=\sum_{i=1}^nX^i_{\alpha}(x)\frac{\partial}{\partial x^i},\qquad\qquad \alpha=1,\ldots,r,$$
are a basis for an $r$-dimensional real  Lie algebra $V$ of vector fields, the so-called associated {\it Vessiot--Guldberg Lie algebra}. In other words, a time-dependent vector field $X(t,x)$ describes a Lie system if and only if the family of vector fields $\{X_t\}_{t\in\mathbb{R}}$, with $X_t:x\in\mathbb{R}^n\rightarrow X_t(x)\equiv X(t,x)\in {\rm T}\mathbb{R}^n$, holds $\{X_t\}_{t\in\mathbb{R}}\subset V$ for a certain finite-dimension Lie algebra of vector fields $V$. 

Various procedures are available to derive a superposition rule. Let us sketch here one of these methods to be used within this work (for a full description, see \cite{CGM07}). The key element of this procedure is the so-called {\it diagonal prolongation} of a vector field. Given a vector field $X(x)=\sum_{i=1}^nX^i(x)\partial/\partial x^i$ over $\mathbb{R}^n$, its diagonal prolongation to $(\mathbb{R}^n)^{p+1}$ is the vector field over this space
\[
\widehat X(x_{0}, \ldots,x_{m})=\sum_{a=0}^p\sum_{i=1}^nX^i(x_{a})\,\frac{\partial}{\partial
x^i_{a}}.
\]
Now, our approach to derive a superposition rule for system (\ref{FORD}) starts by determining a decomposition (\ref{Lievf}) for its associated time-dependent vector field. Next, it must be determined the natural number $m$ so that the {\it diagonal prolongations} to $(\mathbb{R}^n)^m$ of the vector fields $X_1,\ldots,X_r$ become linearly independent at each point of an open dense subset of this space. Note that the diagonal prolongations $\widehat X_1,\ldots,\widehat X_r$ to $\mathbb{R}^{n(m+1)}$ of the vector fields $X_1,\ldots,X_r$ satisfy their same commutation relations and they are linearly independent over an open dense set $U\subset \mathbb{R}^{n(m+1)}$ again. Consequently, they span an involutive distribution 
$$\mathcal{D}_{\bar x}=\langle \widehat X_1(\bar x),\ldots,\widehat X_r(\bar x)\rangle,\qquad \bar x\in U,$$
of rank $r$ over $U$. The vector fields of this distribution admit $n(m+1)-r$ common local first-integrals. Among them, one can choose $n$ first-integrals giving rise to a $n$-codimensional local foliation $\mathcal{F}$ horizontal with respect to the projection $\pi:(x_0,\ldots,x_m)\in(\mathbb{R}^{n})^{m+1}\mapsto (x_1,\ldots,x_m)\in(\mathbb{R}^{n})^m$. Roughly speaking, the former first-integrals enable us to express the coordinates $x_0^i$, with $i=1,\ldots,n,$ in terms of the other variables and the $n$ first-integrals, giving rise to the superposition rule for system (\ref{FORD}).

Since it would be interesting to find a way
to apply the methods for analysing Lie systems to a broader set of differential equations, the theory of quasi-Lie schemes, whose basic features are described below, was developed \cite{CGL08}. Although such a theory applies to systems of first-order differential equations associated with complete and non-complete time-dependent vector fields, the forthcoming presentation will focus on studying systems associated with complete vector fields in order to simplify the presentation and underline the main results of the theory. Nevertheless, a careful analysis shows that our main claims remain valid for systems associated with non-complete vector fields, although with some technical minor modifications.

Let us now turn to define the key notion of the theory of quasi-Lie schemes.

\begin{definition}
A {\it quasi-Lie scheme} $S(W,V)$ is made up by two finite-dimensional vector spaces of vector fields $W,V$ satisfying the following conditions:
\begin{itemize}
 \item $W$ is a linear subspace of $V$.
\item $W$ is a Lie algebra of vector fields, that is, $[W,W]\subset W$.
\item $W$ normalises $V$, i.e. $[W,V]\subset V$.
\end{itemize}
\end{definition}

Roughly speaking, each quasi-Lie scheme $S(W,V)$ determines two families of time-dependent vector fields, $V(\mathbb{R})$ and $W(\mathbb{R})$, which are made of time-dependent vector fields taking values in $V$ and $W$, respectively. The first family, $V(\mathbb{R})$, is intended to describe the first-order systems which can be analysed by means of $S(W,V)$. The second one, $W(\mathbb{R})$ is used to define a group of time-dependent changes of variables: the {\it group of the scheme} $\mathcal{G}(W)$. Now, each system related to a time-dependent vector field of $V(\mathbb{R})$ transforms, under any element of $\mathcal{G}(W)$, into another system related to another element of $V(\mathbb{R})$. This is the most important property of quasi-Lie schemes \cite[Proposition 1]{CGL08}. It has been widely used in recent works \cite{CGL08,CLL08Emd,CLR10Abel} to transform diverse types of differential equations, e.g. Abel equations \cite{CLR10Abel}, into other differential equations of the same type and, more specifically, into Lie systems. In this way, the theory of Lie systems applies to investigate this last system and, by undoing the performed change of variables, multiple properties of the original system can be stated. 
  
Let us explain the above claims more carefully. Every time-dependent vector field $X$ over $\mathbb{R}^n$ gives rise to a {\it generalised flow} $g^X$, i.e. a map $g^X:(t,x)\in\mathbb{R}\times \mathbb{R}^n \mapsto  g^X_t(x)\equiv g^X(t,x)\in \mathbb{R}^n$, with $g^X_0={\rm Id}_{\mathbb{R}^n}$, satisfying that $\gamma^X_{x_0}(t)=g^X_t(x_0)$ is the integral curve of the time-dependent vector field $X$ passing through the point $x_0\in\mathbb{R}^n$ at $t=0$. Now, since $W$ is a Lie algebra of vector fields, it can be proved that the generalised flows corresponding to time-dependent vector fields with values in $W$ form a group, the so-called group of the scheme $\mathcal{G}(W)$, with composition law $(g\star h)_t=g_t\circ h_t$ and $g,h\in\mathcal{G}(W)$. The neutral element is $e:(t,x)\in\mathbb{R}\times\mathbb{R}^n\mapsto x\in\mathbb{R}^n$ and every generalised flow $g$ admits an inverse $g^{-1}:(t,x)\in\mathbb{R}\times\mathbb{R}^n\mapsto (g_t)^{-1}(x)\in\mathbb{R}^n$. As every generalised flow can be considered as a time-dependent change of variables, the group $\mathcal{G}(W)$ can be regarded as a group of time-dependent changes of variables.

Generalised flows can also act on time-dependent vector fields. Given a time-dependent vector field $Y$ and a generalised flow $h$, the action of $h$ over $X$,  let us say $h_\bigstar X$, is the time-dependent vector field  whose generalised flow is $h\star g^Y$. In this terminology, the main result of the theory of quasi-Lie schemes \cite[Proposition 1]{CGL08} establishes that for every time-dependent vector field $Y$ of $V(\mathbb{R})$  and $g\in\mathcal{G}(W)$, the time-dependent vector field $g_\bigstar Y$ belongs to  $V(\mathbb{R})$. In other words, the family of systems associated with the elements of $V(\mathbb{R})$ is stable under the time-dependent changes of variables of $\mathcal{G}(W)$. Among such systems, special relevance have those ones, the so-called {\it Lie systems}, that can be transformed into Lie systems. The precise definition of this notion is detailed below.

\begin{definition}
A system of differential equations describing the integral curves of a time-dependent vector field $X$ is a {\it quasi-Lie system} with respect to a quasi-Lie scheme $S(W,V)$,  if there exist a $g\in\mathcal{G}(W)$ and a Lie algebra $V_0\subset V$ such that $g_\bigstar X$ is a time-dependent vector field taking values in $V_0$.
\end{definition}

Quasi-Lie systems admit certain properties as a consequence of its relation to Lie systems. For instance, each quasi-Lie system admits its general solution to be described in terms of each generic family of particular solutions, a set of constants, and the time, i.e. it admits a {\it time-dependent superposition rule} \cite{CGL08}. Moreover, the theory of quasi-Lie systems provides powerful methods to explicitly determine such superpositions \cite{CGL08,CLL08Emd,CLR10Abel}. In order to understand the relevance of these facts, it is necessary to pay attention to the following remarks.

Every system of first-order differential equations admits time-dependent superposition rules \cite{CGL08,CGL10}. For instance, every first-order system describes the integral curves of a time-dependent vector field admitting a generalised flow which is indeed a particular type of time-dependent superposition rule for the system. In spite of this, determining such time-dependent superposition rules can be, like in the case of the aforementioned example, as difficult as solving the initial system \cite{CGL08, CGL10}. Thus, what it really matters about time-dependent superposition rules is the description of procedures to determine them explicitly as, for instance, quasi-Lie schemes.

Apart from the above remark, there is another reason to use quasi-Lie schemes to determine time-dependent superposition rules: every quasi-Lie scheme provides a family of first-order systems admitting the same time-dependent superposition rule as a bonus \cite[Proposition 14]{CGL10}. Let us briefly analyse this fact. Given a quasi-Lie system with respect to a quasi-Lie scheme $S(W,V)$, there exists an element $g\in \mathcal{G}(W)$ and a Lie algebra $V_0\subset V$ such that the time-dependent vector field $X$ associated with the quasi-Lie system holds that $g_\bigstar X$ takes values in $V_0$. Therefore, the set $S_g(W,V;V_0)$ of quasi-Lie systems with respect to $S(W,V)$ satisfying that $g_\bigstar X$ takes values in $V_0$ is not empty. Moreover, it is made of   time-dependent vector fields of the form $X'=(g^{-1})_\bigstar Y$, with $Y$ being any time-dependent vector field taking values in $V_0$. As all the elements of $S_g(W,V;V_0)$ admit a common time-dependent superposition rule (cf. \cite[Proposition 14]{CGL10}), we have determined a family of systems admitting the same time-dependent superposition rule than $X$.

\section{Superposition rules for systems of SODEs}\label{SRSO}
 
Previous studies about second-order differential equations carried out by means of the theory of Lie systems and quasi-Lie schemes 
were lacking in a theoretical explanation of the methods and notions employed \cite{CGL08,CL08Non,CLR08,RSW97,Ve95}. 
The main aim of this section is to provide a definition for the new notions appearing, with no further explanation, within previous works along with a theoretical explanation of the methods there performed. In this way, a basic background for the posterior theoretical treatment of the subject is laid down.

Recall that the theory of Lie systems was initiated as a result of the study of systems of first-order differential equations admitting their general solutions to be expressed in terms of each generic family of particular solutions and a set of constants. Nevertheless, not only such systems admit their general solutions to be described in this way. For instance, every second-order differential equation of the form $\ddot x=a(t)x$, with $a(t)$ any time-dependent function, satisfies that its general solution, $x(t),$ can be put into the form
\begin{equation}\label{LinearSup}
x(t)=\lambda_1x_{1}(t)+\lambda_{2}x_{2}(t),
\end{equation}
 with $\lambda_1,\lambda_{2}$ being two real constants and $x_{1}(t),x_{2}(t)$ being any family of two particular solutions such that 
 $(x_{1}(t),\dot x_{1}(t))$ and $(x_{2}(t),\dot x_{2}(t))$ are, for every $t\in\mathbb{R}$, two linearly independent elements of ${\rm T}\mathbb{R}\approx \mathbb{R}^2$. In a similar way, other expressions determining the general solution of certain SODEs have recently been described in the literature \cite{CL08Non,CLR08c}. This suggests us to propose the following definition which covers, as particular cases, all the previous expressions occurring in the literature.

\begin{definition}\label{SupRulSec} We say that a system of second-order differential equations 
\begin{equation}\label{SODE}
\ddot x^i=F^i(t,x,\dot x), \qquad \,\, i=1,\ldots,n, 
\end{equation}
on $\mathbb{R}^n$ admits a {\it superposition rule} if there exists a map $\Psi:{\rm T}\mathbb{R}^{mn}\times\mathbb{R}^{2n}\rightarrow \mathbb{R}^n$ such that its general solution, $x(t)$, can be written as
\begin{equation}\label{super}
x(t)=\Psi(x_{1}(t),\ldots,x_{m}(t),\dot x_{1}(t),\ldots,\dot x_{m}(t);\lambda_1,\ldots,\lambda_{2n}),
\end{equation}
in terms of each generic family, $x_{1}(t),\ldots,x_{m}(t),$ of particular solutions, their derivatives, and a set of $2n$ constants.
\end{definition}

In order to grasp the previous definition, it is necessary to precisely establish the meaning of `generic' in the above statement. Formally, it is said that expression (\ref{super}) is valid for a generic family of particular solutions when it holds for every family of particular solutions, $x_{1}(t),\ldots,x_{m}(t),$ satisfying that $(x_{1}(0),\dot x_{1}(0),\ldots,x_{m}(0),\dot x_{m}(0))\in U$, with $U$ being an open dense subset of ${\rm T}\mathbb{R}^{nm}$. In this way, as in the case of superposition rules for Lie systems, the term `generic' amounts to `almost every'.

In order to note that all those aforementioned expressions studying the general solutions for certain SODEs are superposition rules in the above sense, it is necessary to explain an important detail. Some of such expressions, like the one for studying Milne--Pinney equations, depend on a generic set of $m$ particular solutions, a set of constants and a set of time-independent constants of the motion \cite{CL08Non}. These expressions seem to differ from the above definition. Nevertheless, if we take into account that such constants of the motion depend on the $m$ particular solutions, we can notice that, indeed, they are superposition rules in the above sense.

Although there exists no characterisation for systems of SODEs of the form (\ref{SODE}) admitting a superposition rule, there exists a special class of such systems, the so-called {\it SODE Lie systems}, accepting such a property. Even though this fact has been employed implicitly in the literature, it has never been proved explicitly. In view of these facts, we next furnish the definition of the SODE Lie system notion along with a proof showing that every SODE Lie system admits a superposition rule. In addition, some remarks about the interest of this notion and its main properties are discussed.

\begin{definition}\label{DefSODE} We say that the system of  SODEs (\ref{SODE})
is a {\it SODE Lie system} if the system of first-order differential equations
\begin{equation}\label{FOrder}
\left\{
\begin{aligned}
\dot x^i&=v^i,\\
\dot v^i&=F^i(t,x,v),
\end{aligned}\right.\qquad i=1,\ldots,n,
\end{equation}
obtained by adding the new variables $v^i=\dot x^i$, with $i=1,\ldots,n$, to system (\ref{SODE}), is a Lie system.
\end{definition}

\begin{proposition}\label{SR}  Every SODE Lie system (\ref{SODE}) admits a superposition rule $\Psi:{\rm T}\mathbb{R}^{nm}\times\mathbb{R}^{2n}\rightarrow\mathbb{R}^n$ of the form $\Psi=\pi\circ\Phi$, where $\Phi:{\rm T}\mathbb{R}^{nm}\times\mathbb{R}^{2n}\rightarrow {\rm T}\mathbb{R}^n$ is a superposition rule for the system (\ref{FORD}) and $\pi:{\rm T}\mathbb{R}^n\rightarrow\mathbb{R}^n$ is the projection associated with the tangent bundle ${\rm T}\mathbb{R}^n$. 
\end{proposition}
\begin{proof}
Each SODE Lie system of the form (\ref{SODE}) is associated with a first-order system of differential equations (\ref{FOrder}) admitting a superposition rule $\Phi:{\rm T}\mathbb{R}^{nm}\times \mathbb{R}^{2n}\rightarrow {\rm T}\mathbb{R}^n$. This allows us to describe the general solution $(x(t),v(t))$ of system (\ref{FOrder}) in terms of each generic set $(x_{a}(t),v_{a}(t))$, with $a=1,\ldots,m$, of its particular solutions and a set of $2n$ constants, i.e.
\begin{equation}\label{SupRel2}
\begin{aligned}
(x(t), v(t))&=\Phi\left(x_{1}(t),\ldots, x_{m}(t),v_{1}(t),\ldots, v_{m}(t);\lambda_1,\ldots,\lambda_{2n}\right)\\
\end{aligned}.
\end{equation}
Obviously, each solution, $x_p(t)$, of the second-order system (\ref{SODE}) corresponds to one and only one solution $(x_p(t),v_p(t))$ of (\ref{FORD}) and viceversa. Since $(x_p(t),v_p(t))=(x_p(t),\dot x_p(t))$, it follows that the general solution $x(t)$ of (\ref{SODE}) can be put in the form
\begin{equation}\label{SupRel4}
x(t)=\pi\circ\Phi\left(x_{1}(t),\ldots, x_{m}(t),\dot x_{1}(t),\ldots, \dot x_{m}(t);\lambda_1,\ldots,\lambda_{2n}\right),
\end{equation}
where $x_{a}(t)$, with $a=1,\ldots,n$, is a generic family of particular solutions of (\ref{SODE}). In other words, the map $\Psi=\pi\circ\Phi$ is a superposition rule for the system of second-order differential equations (\ref{SODE}).
\end{proof}

Since every autonomous system is related to a one-dimensional Vessiot--Guldberg Lie algebra \cite{CGL08}, it straightforwardly follows, from the above proposition, next corollary.

\begin{corollary} Every autonomous system of second-order differential equations of the form $\ddot x^i=F^i(x,\dot x)$, with $i=1,\ldots,n$, admits a superposition rule. 
\end{corollary}

Despite its theoretical interest, the above result cannot straightforwardly be used to derive superposition rules most of the times. Actually, the superposition rule
 guaranteed by Proposition \ref{SR} relies on obtaining a superposition rule for an autonomous first-order system 
 of differential equations. Considering the method sketched during Section \ref{LSLS}, we find that determining 
 this superposition rule demands the determination of all the integral curves of a vector field on $({\rm T}\mathbb{R}^n)^2$. Although the existence
  of the solution of this problem is known, its explicit description can be as difficult as solving the initial system (indeed, this is usually the case). Consequently, deriving explicitly a superposition rule for the above autonomous system frequently relies on searching an alternative superposition rule for the associated first-order system.

Many superposition rules for systems of second-order differential equations do not show an explicit dependence on the derivatives of the particular solutions. Consider, for instance,  either the linear superposition rule (\ref{LinearSup}) for the equation $\ddot x=a(t)x$,  or the affine one, $$x(t)=\lambda_1(x_{1}(t)-x_{2}(t))+\lambda_2(x_{2}(t)-x_{3}(t))+x_{3}(t),$$
for $\ddot x=a(t)x+b(t)$.  Such superposition rules are called {\it velocity free superposition rules} or even {\it free superposition rules}. The conditions ensuring the existence of such superposition rules is an interesting open problem. 

\begin{proposition}
Every system of SODEs (\ref{SODE}) admitting a free superposition rule is a SODE Lie system.
\end{proposition}
\begin{proof}

Suppose that system (\ref{SODE}) admits a superposition rule of the special form
\begin{equation}\label{FreeSuperRule}
\begin{aligned}
x^i&=\Phi_x^i(x_{1},\ldots, x_{m};\lambda_1,\ldots,\lambda_{2n}),
\end{aligned}\qquad i=1,\ldots,n.
\end{equation}
In such a case, the general solution, $x(t)$, of the system could be expressed as
\begin{equation}\label{freeSup1}
x^i(t)=\Phi^i_x(x_{1}(t),\ldots,x_{m}(t);\lambda_1,\ldots,\lambda_{2n}), \qquad i=1,\ldots,n.
\end{equation}
If we set $p(t)=(x_{1}(t),\ldots,x_{m}(t),\dot x_{1}(t),\ldots,\dot x_{m}(t))$, $v^i=\dot x^i$, and $v_a^i=\dot x^i_a$ with $a=1,\ldots,m$, the derivative of the above expression with respect to $t$ reads
\begin{equation}\label{freeSup2}
v^i(t)=\dot x^i(t)=\sum_{a=1}^m\sum_{j=1}^n\left(v_{a}^j(t)\frac{\partial\Phi_x^i}{\partial x^j_{a}}(p(t))\right),\qquad i=1,\ldots,n.
\end{equation}
Consequently, there exists a function $$\Phi^i_v(x_1,\ldots,x_m,v_1,\ldots,v_m)=\sum_{a=1}^m\sum_{j=1}^n\left(v_{a}^j\frac{\partial\Phi_x^i}{\partial x^j_{a}}\right),\qquad i=1,\ldots,n,
$$
such that 
\begin{equation*}
\left\{
\begin{aligned}
x^i(t)&=\Phi_x^i(x_{1}(t),\ldots, x_{m}(t);\lambda_1,\ldots,\lambda_{2n}),\\
v^i(t)&=\Phi_v^{i}(x_{1}(t),\ldots, x_{m}(t),v_{1}(t),\ldots, v_{m}(t);\lambda_1,\ldots,\lambda_{2n}),
\end{aligned}\right.\qquad i=1,\ldots,n.
\end{equation*}
Therefore, system (\ref{FOrder}) admits a  superposition rule and (\ref{SODE}) becomes a SODE Lie system.
\end{proof}

Apart from the SODE Lie system notion, there exists another method to study certain systems of second-order differential equations which admit a regular Lagrangian, like Caldirola--Kanai oscillators or Milne--Pinney equations \cite{CLR08,Ru10}. Although this method cannot be used to study general second order systems, it provides us with some additional information that cannot be derived by means of SODE Lie systems when it applies, e.g. about the time-dependent constants of the motion of the system \cite{Ru10}. 

A possible generalisation of the concept of superposition rule appearing, for instance, in the theory of quasi-Lie schemes is the so-called {\it time-dependent superposition rule} \cite{CGL08}. This concept can be extended to the framework of systems of second-order differential equations as follows.

\begin{definition} We say that the map $\Psi:\mathbb{R}\times {\rm T}\mathbb{R}^{mn}\times\mathbb{R}^{2n}\rightarrow \mathbb{R}^n$ is a {\it time-dependent superposition rule} for the system of SODEs (\ref{SODE}), if its general solution $x(t)$ can be written in terms of each generic family $x_{(1)}(t),\ldots,x_{(m)}(t)$ of particular solutions, their derivatives, a set of $2n$ constants, and the time as
$$
x(t)=\Psi(t,x_{1}(t),\ldots,x_{m}(t),\dot x_{1}(t),\ldots,\dot x_{m}(t);\lambda_1,\ldots,\lambda_{2n}).
$$
\end{definition}

It is essential to analyse the existence of time-dependent superposition rules in order to understand the relevance of the practical results carried out throughout this work. As it shall be shown soon, many of the properties of these superpositions are a consequence of the features of time-dependent superposition rules for first-order systems. For instance, let us prove the following results concerning the existence of time-dependent superposition rules for systems of SODEs.

\begin{proposition}\label{GP} Every system of SODEs (\ref{SODE}) admits a time-dependent superposition rule of the form $\Psi:\mathbb{R}\times\mathbb{R}^{2n}\rightarrow\mathbb{R}^n$. 
\end{proposition}
\begin{proof}
Every system (\ref{SODE}) is related to a first-order system (\ref{FORD}) admitting a flow $g:(t;\lambda)\in\mathbb{R}\times\mathbb{R}^{2n}\mapsto g_t(\lambda)\in{\rm T}\mathbb{R}^n$ which allows us to cast its general solution, $\xi(t)$, into the form $\xi(t)=g_t(k)$. Consequently, the general solution, $x(t)$, of system (\ref{SODE}) can be written as $x(t)=\pi\circ g_t(k)$. In other words, system (\ref{SODE}) admits a time-dependent superposition rule depending just on $2n$ constants. 
\end{proof}

Despite the remarkable theoretical interest of the above result, it does not provide any additional method for the explicit derivation of solutions of systems of SODEs. Indeed, note that the derivation of the above time-dependent superposition rule amounts to working out the generalised flow for the first-order system (\ref{FORD}). This involves solving the system for each initial condition. If this can explicitly be done, determining the above superposition rule becomes unnecessary; and, otherwise, the superposition rule, although interesting, cannot be provided. 

Apart from the time-dependent superposition rule guaranteed by Proposition \ref{GP}, other instances can be ensured to exist. Nevertheless, their explicit determination uses to be as difficult as solving the initial system. Let us illustrate this statement  more carefully. Consider a system of SODEs (\ref{SODE}) related to the system of first-order differential equations (\ref{FOrder}), with general solution $p_x(t)=(x(t),v(t))$. Let $X$ be the time-dependent vector field associated with this system, and $Y$ any other time-dependent vector field on ${\rm T}\mathbb{R}^n\simeq\mathbb{R}^{2n}$. Their corresponding flows, $g^X,g^Y:\mathbb{R}\times{\mathbb{R}^{2n}}\rightarrow{\rm T}\mathbb{R}^n$, satisfy that $p_x(t)=g^X_t\circ (g^Y_t)^{-1}(p_y(t))$, where $p_y(t)$ is the general solution of the system describing the integral curves of $Y$. Therefore, $(g^X\circ (g^Y)^{-1})_\bigstar Y=X$ (cf. \cite{CGL08}). In particular, if $Y=0$, then $p_y(t)=\phi(\lambda_1,\ldots,\lambda_{2n})$, where $\phi:\mathbb{R}^{2n}\rightarrow{\rm T}\mathbb{R}^n$ is any diffeomorphism. Therefore, 
$$
p_x(t)=g_t^X(\lambda_1,\ldots,\lambda_{2n})\Longrightarrow x(t)=\pi\circ g_t^X(\lambda_1,\ldots,\lambda_{2n}).
$$
In other words, system (\ref{SODE}) admits a time-dependent superposition rule depending just on a set of $2n$ constants. In a similar way, if we assume $Y$ to be any autonomous, non-null, vector field, the system describing its integral curves is a Lie system ($[Y,Y]=0$) and the straightforward application of the method developed in \cite{CGL08} shows that it admits a superposition rule depending on one particular solution $p_{y_1}(t)=(y_1(t),\dot y_1(t))$, i.e. $p_y(t)=\Phi(p_{y_{1}}(t);\lambda_1,\ldots,\lambda_{2n})$. Consequently, the general solution, $x(t)$, of the system (\ref{SODE}) associated with $X$ now reads
\begin{equation*}
x(t)=\pi\circ g^X_t\circ(g^Y)^{-1}_t\circ  \Phi(g^Y_t\circ(g^{X}_t)^{-1}(p_{x_{1}}(t));\lambda_1,\ldots,\lambda_{2n}),\\
\end{equation*}
where $p_{x_1}(t)=(x_1(t),\dot x_1(t))$ is a particular integral curve of $X$. That is, system (\ref{SODE}) admits a time-dependent superposition rule depending on one particular solution.

Note that the determination of the above superposition rules relies, among other things, on the determination of the flow of the initial time-dependent vector field. Therefore, obtaining the superposition rule for (\ref{SODE}) is as difficult as solving the system (\ref{SODE}). Other similar constructions can be found. Nevertheless, most of them share the same drawbacks.

The above remarks make evident that a key point in the study of time-dependent superposition rules for systems of SODEs is the
development of procedures to explicitly determine them. As it occurred in the explicit determination of time-dependent superposition rules for systems of first-order differential equations, quasi-Lie schemes play an important role in the description of superpositions for systems of SODEs. Indeed, the same commentaries pointed out at the end of Section \ref{LSLS} can also be claimed here. 

On one hand, given a system of SODEs (\ref{SODE}), quasi-Lie schemes provide a powerful tool of large applicability to derive time-dependent superposition rules for its associated first-order system (\ref{FOrder}). From here, it is immediate to obtain a time-dependent superposition rule for (\ref{SODE}).

On the other hand, quasi-Lie schemes naturally provide a family of first-order systems, including (\ref{FOrder}), whose elements  admit the same obtained time-dependent superposition rule. It is easy to prove that all those systems of SODEs whose associated first-order systems are members of this family share the same time-dependent superposition rule. In Section \ref{QLSORE} this fact will be clarified and illustrated through the study of second-order Riccati equations.

\section{A new superposition rule for a family of SODE Lie systems}

In this Section we derive a superposition rule for a family of second-order differential equations including, as particular instances, some Painlev\'e--Ince equations \cite{EEU07}. In the process of searching for such a superposition rule, we find a family of Lie systems which will be used posteriorly to describe time-dependent and time-independent superposition rules for other second-order differential equations studied in Physics and Mathematics.

Consider the family of differential equations 
\begin{equation}\label{MDPIeq}
\ddot x+3x\dot x+x^3=f(t),
\end{equation}
with $f(t)$ being any time-dependent function. The interest in these equations is motivated by their frequent appearance in the Physics and Mathematics literature \cite{CRS05,CLS05II,KL09}. The  properties of these equations have been deeply analysed since their first analysis by Vessiot and Wallenberg \cite{Ve94,Wall03} as a particular case of second-order Riccati equations. For instance, these equations appear in  \cite{GL99} in the study of
 Riccati chain. In that work it is stated that the equations of the Riccati chain can be used to derive solutions for certain PDEs. In addition, equation (\ref{MDPIeq}) also appears in the book by Davis \cite{Da62}, and the particular case with $f(t)=0$  has recently been treated through geometric methods in \cite{CGR09,CRS05}.

The results described in previous sections can be used to study differential equations (\ref{MDPIeq}). Let us first  show that the above differential equations are SODE Lie systems and, in view of Proposition 1, they admit a superposition rule that is derived.  According to  definition \ref{SODE}, equation (\ref{MDPIeq}) is a SODE Lie system if and only if the system 
\begin{equation}\label{FO}
\left\{\begin{aligned}
\dot x&=v,\\
\dot v&=-3xv-x^3+f(t),
\end{aligned}\right.
\end{equation}
determining the integral curves of the time-dependent
vector field of the form 
\begin{equation}\label{Dec}
X_{PI}(t,x,v)=X_1(x,v)+f(t)X_2(x,v), 
\end{equation}
with
$$
X_1=v\frac{\partial}{\partial x}-(3xv+x^3)\frac{\partial}{\partial v},\qquad X_2=\frac{\partial}{\partial  v},
$$
is a Lie system.

In view of the decomposition (\ref{Dec}), all equations (\ref{MDPIeq}) are SODE Lie systems if the vector fields $X_1$ and $X_2$ are included in a finite-dimensional real  Lie algebra of vector fields $V$. This happens  if and only if $X_1$, $X_2$ and all their successive Lie brackets, i.e. the vector fields of the form 
\begin{equation}\label{Envelope}
[X_1,X_2], [X_1,[X_1,X_2]], [X_2,[X_1,X_2]], [X_1,[X_1,[X_1,X_2]]], \ldots
\end{equation}
span a finite-dimensional Lie algebra. Consider the family of vector fields on ${\rm T}\mathbb{R}$ given by {\small
\begin{equation}\label{VF}
\begin{aligned}
X_1&=v\frac{\partial}{\partial x}-(3xv+x^3)\frac{\partial}{\partial v},\,\, &X_2&=\frac{\partial}{\partial  v},\\
X_3&=-\frac{\partial}{\partial x}+3x\frac{\partial}{\partial v},\,\, &X_4&=x\frac{\partial}{\partial x}-2x^2\frac{\partial}{\partial v},\\
X_5&=(v+2x^2)\frac{\partial}{\partial x}-x(v+3x^2)\frac{\partial}{\partial v},\,\, &X_6&=2x(v+x^2)\frac{\partial}{\partial  x}+2(v^2-x^4)\frac{\partial}{\partial v},\\
X_7&=\frac{\partial}{\partial x}-x\frac{\partial}{\partial v},\,\,
&X_8&=2x\frac{\partial}{\partial x}+4v\frac{\partial}{\partial v},
\end{aligned}
\end{equation}}where $X_3=[X_1,X_2]$, $-3 X_4=[X_1,X_3]$, $X_5=[X_1,X_4]$, $X_6=[X_1,X_5]$, $X_7=[X_2,X_5]$, $X_8=[X_2,X_6]$. Then, the vector fields $X_1,\ldots, X_8$ are linearly independent  over $\mathbb{R}$. Additionally, in view of the previous commutation relations and 
{\small\begin{equation}\label{Rel}
\begin{array}{lllll}
\left[X_1,X_6\right]=0, &[X_1,X_7]=\frac 12 X_8,&\left[X_1,X_8\right]=-2X_1,&[X_2,X_3]=0,\\
\left[X_2,X_4\right]=0, &[X_2,X_7]=0,           &[X_2,X_8]=4X_2,            &[X_3,X_4]=-X_7,\\
\left[X_3,X_5\right]=-\frac 12X_8,&[X_3,X_6]=-2X_1,&[X_3,X_7]=-2X_2,&[X_3,X_8]=2X_3,\\
\left[X_4,X_5\right]=-X_1,&[X_4,X_6]=0,         &[X_4,X_7]=X_3,             &[X_4,X_8]=0,\\
\left[X_5,X_6\right]=0,   &[X_5,X_7]=-3X_4,&[X_5,X_8]=-2X_5,&[X_6,X_7]=-2X_5,\\
&\left[X_6,X_8\right]=-4X_6,&[X_7,X_8]=2X_7,&\\
\end{array}
\end{equation}}it follows that the vector fields $X_1,\ldots,X_8$ span an eight-dimensional Lie algebra of vector fields $V$ containing $X_1$ and $X_2$.
 Therefore, equation  (\ref{MDPIeq}) is a SODE Lie system.  Moreover,  the elements of the following  family of traceless real $3\times 3$ matrices
\begin{gather*}
M_1=-\left(
\begin{array}{ccc}
0&1&0\\
0&0&1\\
0&0&0.
\end{array}\right),
M_2=-\left(
\begin{array}{ccc}
0&0&0\\
0&0&0\\
1&0&0.
\end{array}\right),
M_3=\left(
\begin{array}{ccc}
0&0&0\\
1&0&0\\
0&-1&0.
\end{array}\right),\\
M_4=\frac{1}{3}\left(
\begin{array}{ccc}
1&0&0\\
0&-2&0\\
0&0&1.
\end{array}\right),
M_5=\left(
\begin{array}{ccc}
0&1&0\\
0&0&-1\\
0&0&0.
\end{array}\right),
M_6=\left(
\begin{array}{ccc}
0&0&2\\
0&0&0\\
0&0&0.
\end{array}\right),\\
M_7=-\left(
\begin{array}{ccc}
0&0&0\\
1&0&0\\
0&1&0.
\end{array}\right),
M_8=\left(
\begin{array}{ccc}
2&0&0\\
0&0&0\\
0&0&-2.
\end{array}\right),
\end{gather*}
satisfy the same commutation relations as the 
vector fields $X_1,\ldots,X_8$, i.e. the linear map  $\rho:\mathfrak{sl}(3,\mathbb{R})\rightarrow V$,
such that $\rho(M_\alpha)=X_\alpha$, with $\alpha=1,\ldots,8$, is a
Lie algebra isomorphism. Consequently, the finite-dimensional Lie algebra of vector
fields $V$ is isomorphic to $\mathfrak{sl}(3,\mathbb{R})$ and the systems of differential equations describing the integral curves for the time-dependent vector fields 
\begin{equation}\label{family}
X(t,x,v)=\sum_{\alpha=1}^8b_\alpha(t)X_\alpha(x,v),
\end{equation}
are Lie systems related to a Vessiot--Guldberg Lie algebra isomorphic to $\mathfrak{sl}(3,\mathbb{R})$.

Recall that, in view of Proposition \ref{SR}, a superposition rule for (\ref{MDPIeq}) can be obtained by means of a superposition rule for the Lie system (\ref{FO}). As we already stated in Section \ref{LSLS}, a superposition rule for a
certain Lie system describing the integral curves of a time-dependent vector field $X$ admitting a decomposition of the form (\ref{family}) can be obtained through the first-integrals for the distribution $\mathcal{D}$ spanned by the diagonal prolongations $\widehat X_1\ldots, \widehat X_8$ on a certain ${\rm T}\mathbb{R}^{n(m+1)}$ in such a way that their projections $\pi_*(\widehat X_\alpha)$, with $\alpha=1,\ldots,8$, are linearly independent on a dense open subset of ${\rm T}\mathbb{R}^{nm}$.

Consider the functions $F_{abc}$, with $a,b,c=0,\ldots,4$, given by
\begin{equation}\label{Fabc}
F_{abc}=v_a(x_c-x_b)+v_b(x_a-x_c)+v_c(x_b-x_a)+(x_a-x_b)(x_b-x_c)(x_c-x_a).
\end{equation}
Such functions satisfy that $F_{abc}=F_{bca}=F_{cab}=-F_{bac}=-F_{cba}=-F_{acb}$, for all $a,b,c=0,\ldots,4$, and they are useful to determine the superposition rule we are looking for. 
If we take now the diagonal prolongations to ${\rm T}\mathbb{R}^{5}$ of
the family of vector fields (\ref{VF}), we can get, e.g. by means of any symbolic manipulation program, that the vector fields $\pi(\widetilde X_\alpha)$, with $\alpha=1,\ldots,8$, on ${\rm T}\mathbb{R}^4$ are linearly independent at those points $p\equiv(x_1,\ldots,x_4,v_1,\ldots,v_4)\in{\rm T}\mathbb{R}^4$ satisfying $F_{123}(p)F_{124}(p)F_{134}(p)F_{234}(p)\neq 0$. 
Such a set is dense in ${\rm T}\mathbb{R}^4$ and hence the involutive distribution $\mathcal{D}$ spanned by $\widehat X_1\ldots,\widehat X_8$ on ${\rm T}\mathbb{R}^5$ is eight-dimensional almost everywhere. Additionally, there exist, at least locally, two first-integrals for the vector fields of the distribution $\mathcal{D}$ which can be used to derive a superposition rule. 

As the vector fields $\widehat X_1$, $\widehat X_2$ and their successive Lie brackets, i.e. the vector fields (\ref{Envelope}), span the distribution $\mathcal{D}$, it can be proved that giving a first-integral
$F:{\rm T}\mathbb{R}^{5}\rightarrow\mathbb{R}$ for the vector fields  $\widehat X_1$ and $\widehat X_2$, i.e. $\widehat X_1F=\widehat
X_2F=0$, is equivalent to giving a first-integral for the distribution $\mathcal{D}$.

These integrals can be obtained by applying the method of characteristics to the vector fields $\widehat X_1$ and $\widehat X_2$. The calculation of these first-integrals is long and it is detailed in the Appendix. As a result, we get that two first-integrals for $\widehat X_1$ and $\widehat X_2$ and hence for all the vector fields of the distribution $\mathcal{D}$ are
$$
\Lambda_1(p)=\frac{F_{431}(p)F_{210}(p)}{F_{421}(p)F_{310}(p)}\qquad {\rm and}\qquad \Lambda_2(p)=\frac{F_{431}(p)F_{420}(p)}{F_{421}(p)F_{430}(p)}
$$
where now $p\equiv(x_0,x_1,\ldots,x_4,v_0,v_1,\ldots,v_4)\in{\rm T}\mathbb{R}^4$.  We can obtain $v_0$ from the expression of $\Lambda_1$ as
\begin{multline}\label{vExp}
v_0=\frac{(v_1(x_2-x_0)+v_2(x_0-x_1)+(x_1-x_0)(x_0-x_2)(x_2-x_1))F_{431}}{(x_2-x_1)F_{431}+(x_1-x_3)F_{421}\Lambda_1}+\\ \frac{(v_3(x_1-x_0)+v_1(x_0-x_3)+(x_0-x_1)(x_1-x_3)(x_3-x_0))F_{421}\Lambda_1}{(x_2-x_1)F_{431}+(x_1-x_3)F_{421}\Lambda_1},
\end{multline}
and if we substitute this value in the expression for $\Lambda_2$, we obtain the expression 
{\small
\begin{equation}\label{Sup}
x_0=\frac{x_2F_{431}-G_{3124}\Lambda_2-G_{2134}\Lambda_1+x_3F_{421}\Lambda_1\Lambda_2}{F_{431}+(F_{124}-F_{324})\Lambda_1+(F_{412}-F_{312})\Lambda_2+\Lambda_1\Lambda_2F_{421}},
\end{equation}
}where
\begin{multline*}
G_{abcd}=x_a((v_d-v_c)x_b+(v_b-v_d)x_c+(x_b-x_c)x_bx_c+(x_c-x_b)x_ax_d)+\\x_d((v_c-v_a)x_b+(v_a-v_b)x_c+(x_c-x_b)x_bx_c+(x_b-x_c)x_ax_d).
\end{multline*}
Note that the first-integrals $\Lambda_1$ and $\Lambda_2$ satisfy that 
$$
\Lambda_j(x_{0}(t),\ldots,x_{4}(t),v_{0}(t),\ldots,v_{4}(t))={\rm cte.} \qquad j=1,2,
$$
where $(x_{a}(t),v_{a}(t))$, with $a=0,\ldots,4$, are any family of five particular solutions of system (\ref{FO}). Consequently, if we fix $x_{0}(t)=x(t)$, $v_{0}(t)=\dot x_{0}(t)$ and take into account expression (\ref{Sup}), we see that the general solution of system (\ref{FO}), namely, $x(t)$, can be written in terms of each generic family of four particular solutions $x_{a}(t)$, with $a=1,\ldots,4$, their derivatives in terms of $t$, and two real constants $\lambda_1$, $\lambda_2$ as follows
{\small
\begin{equation}\label{FinSup}
x(t)=\frac{x_2(t)F_{431}(t)-G_{3124}(t)\lambda_2-G_{2134}(t)\lambda_1+x_3(t)F_{421}(t)\lambda_1\lambda_2}{F_{431}(t)+(F_{124}(t)-F_{324}(t))\lambda_1+(F_{412}(t)-F_{312}(t))\lambda_2+\lambda_1\lambda_2F_{421}(t)},
\end{equation}}where the time-dependent functions $F_{abc}(t)$ and $G_{abcd}(t)$ are obtained evaluating the functions $F_{abc}$ and $G_{abcd}$ on the curves $(x_{1}(t),\ldots,x_{4}(t),\dot x_{1}(t),\ldots,\dot x_{4}(t))$.

In order to illustrate the previous result by means of a simple example, let us consider, for instance, equation (\ref{MDPIeq}) with $f(t)=0$. By direct inspection, we can get the set of particular solutions
\begin{equation}\label{PS}
x_{1}(t)=0,\qquad x_{2}(t)=\frac 2t,\qquad x_3(t)=\frac{2t}{2+t^2},\qquad x_4(t)=\frac{1+2t}{t+t^2}.
\end{equation}
These particular solutions can be used to determine the functions appearing in the expression (\ref{Sup}), i.e. $G_{2134}$, $F_{431}$, etc. More specifically, we get, in view of the particular solutions (\ref{PS}), that 
$$
G_{3124}(t)=\frac{2t^{-2}}{(t^2+1)(t+1)},\,\,  F_{431}(t)=\frac{2t^{-1}}{(t^2+1)(t+1)},\,\, G_{2134}(t)=\frac{-4t^{-1}}{(t^2+1)(t+1)}$$
and 
$$F_{124}(t)=\frac{2}{t^2+t^3},\quad F_{324}(t)=\frac{2}{t^2+t^3+t^4+t^5},\quad F_{312}(t)=\frac{2}{2t+t^3}.
$$
Finally, making use of expression (\ref{Sup}) and the above functions, we get that the general solution for equation (\ref{MDPIeq}) is
$$
x(t)=\frac{(1+2t\lambda_1)(-1+\lambda_2)}{t(-1+\lambda_2)+t^2\lambda_1(-1+\lambda_2)+(-1+\lambda_1)\lambda_2}.
$$

In view of Proposition \ref{SR}, expression (\ref{FinSup}) is not only a superposition rule for the equation (\ref{MDPIeq}), but also for any SODE Lie system (\ref{SODE}) satisfying that its corresponding system (\ref{FOrder}) is related to a time-dependent vector field that put into a form similar to (\ref{family}), i.e. it takes values in $V$. 
Many instances of the family of Lie systems (\ref{family}) are associated with interesting SODE Lie systems with applications to Physics or related to interesting mathematical problems. In all these cases, the theory of Lie systems can be applied to investigate these second-order differential equations, recover some of their known properties, and, possibly, provide new results. Let us illustrate this assertion by means of a few examples.

Another equation appearing in the Physics literature \cite{CLS05II,CLS05,TT07} which can be analysed by means of our methods is 
\begin{equation}\label{Exam2}
\ddot x+3x\dot x +x^3+\lambda_1x=0,
\end{equation}
which is a special kind of Li\'enard equation $\ddot x+f(x)\dot x+g(x)=0$, with $f(x)=3x$ and $g(x)=x^3+\lambda_1x$. The above equation can also be related to a generalised form of an Emden equation occurring in the thermodynamical study of equilibrium configurations of spherical clouds of gas acting under the mutual attraction of their molecules \cite{DT90}.

As in the study of equations (\ref{MDPIeq}), by considering the new variable $v=\dot x$, equation (\ref{Exam2}) becomes the system
\begin{equation}
\left\{\begin{aligned}
\dot x&=v,\\
\dot v&=-3xv-x^3-\lambda_1x,
\end{aligned}\right.
\end{equation}
describing the integral curves of the vector field $X=X_1-\lambda_1/2(X_7+X_3)$ included in the family (\ref{family}). Consequently, the expression (\ref{Sup}) can be used to derive the solution of equation (\ref{Exam2}) in terms of a set of four particular solutions.

Finally, we can also treat by our methods the equation  
\begin{equation}\label{general}
\ddot x+3x\dot x +x^3+f(t)(\dot x+x^2)+g(t)x+h(t)=0,
\end{equation}
containing, as particular cases, all the previous examples \cite{KL09}.
The system of first-order differential equations associated with this equation reads
\begin{equation}\label{FirstGeneral}
\left\{\begin{aligned}
\dot x&=v,\\
\dot v&=-3xv-x^3-f(t)(v+x^2)-g(t)x-h(t).
\end{aligned}\right.
\end{equation}
Hence, this system describes the integral curves of the time-dependent vector field 
$$X_t=X_1-h(t)X_2-\frac 14 f(t)\,(X_8-2X_4)-\frac 12 g(t)\,(X_7+X_3).$$
Therefore, equation (\ref{general}) is a SODE Lie system and the theory of
Lie systems can be used to analyse its properties. In particular, the expression (\ref{Sup}) provides us with 
  the general solutions for these equations out of each generic set of four particular solutions. 

Some particular cases of system
(\ref{general}) were pointed out in \cite{CLS05,KL09}. Additionally, the case with
$f(t)=0$, $g(t)=\omega^2(t)$ and $h(t)=0$ was studied in
 \cite{CLS05II} and it is related to harmonic oscillators. The case with $g(t)=0$ and $h(t)=0$ appears in the catalogue of equations possessing the Painlev\'e property \cite{In86}. Finally, our result generalises Vessiot's result \cite{Ve95} describing the existence of an expression determining the general solution of system  like (\ref{general}) (but with constant coefficients) in terms of four of their particular solutions, their derivatives and two constants.

\section{Quasi-Lie schemes and second-order Riccati equation}\label{QLSORE}
In this Section we derive a time-dependent superposition rule for the second-order Riccati equation \cite{CRS05}
\begin{equation}\label{NLe}
\ddot x+\left(b_0(t)+b_1(t)x\right)\dot x+a_0(t)+a_1(t)x+a_2(t)x^2+a_3(t)x^3=0,
\end{equation}
with $a_3(t)>0$, $a_3(0)=1$, $b_0(t)=\frac{a_2(t)}{\sqrt{a_3(t)}}-\frac{\dot a_3(t)}{2a_3(t)}$ and $b_1(t)=3\sqrt{a_3(t)}$,  by means of the theory of quasi-Lie schemes. 

We first introduce a new variable 
$v=\dot x$ and transform
equation (\ref{NLe}) into the system of first-order differential equations
\begin{equation}\label{FOSOR}
\left\{
\begin{aligned}
\dot x&=v,\\
\dot v&=-\left(b_0(t)+b_1(t)x\right)v-a_0(t)-a_1(t)x-a_2(t)x^2-a_3(t)x^3.
\end{aligned}\right.
\end{equation}
Consider the following set of vector fields
\begin{align*}
Y_1=v\frac{\partial}{\partial x},\,\,\,\, Y_2=v\frac{\partial}{\partial v},\,\,\,\, Y_3=xv\frac{\partial}{\partial v},\,\,\,\, Y_4=
\frac{\partial}{\partial v},\,\,\,\,\\ 
Y_5=x\frac{\partial}{\partial v},\,\,\,\, Y_6=x^2\frac{\partial}{\partial v},\,\,
\,\, Y_7=x^3\frac{\partial}{\partial v},\,\,\,\, Y_8=x\frac{\partial}{\partial x},
\end{align*}
spanning a linear space of vector fields $V=\langle Y_1,\ldots,Y_8\rangle$. The 
system (\ref{FOSOR}) describes the integral curves of the time-dependent vector field 
$$
Y_t=Y_1-b_0(t)Y_2-b_1(t)Y_3-a_0(t)Y_4-a_1(t)Y_5-a_2(t)Y_6-a_3(t)Y_7,
$$
and, therefore, as $Y_t\in V$, for every $t\in\mathbb{R}$, we get that $Y\in V(\mathbb{R})$. 
Let us denote ${\rm ad}_Z(X)=[Z,X]$, for every vector fields $X,Z$. Hence, as 
$${\rm ad}_{Y_3}^n(Y_6)=\overbrace{{\rm ad}_{Y_3}\circ\ldots\circ{\rm ad}_{Y_3}}^{n\,{\rm times}}(Y_6)=(-x)^{n+2}\frac{\partial}{\partial v},$$
each Lie algebra of vector fields $V'$ containing the vector fields $Y_3$ and $Y_6$ must include the above infinite family of linearly independent vector fields over $\mathbb{R}$. Consequently, as $Y_3$ and $Y_6$ are contained in $V$, there exists no finite-dimensional Lie algebra of vector fields $V'$ including $V$ and, therefore, equation (\ref{FOSOR}) is not a Lie system.

However, we can deal with such a system as a   quasi-Lie system with respect to a quasi-Lie scheme \cite[Theorem 4]{CGL08} and therefore 
this could be used to build up a time-dependent superposition rule for equation (\ref{NLe}). 
In fact, define the linear space $W=\langle Y_2,Y_8\rangle$, which  is a two-dimensional Abelian Lie algebra of
vector fields and, in view of the commutation relations
$$
\begin{aligned}
\left[Y_2,Y_1\right]&=Y_1,\,\, &[Y_2,Y_3]&=0,\,\, &[Y_2,Y_4]&=-Y_4,&\left[Y_2,Y_5\right]&=-Y_5,\\
\left[Y_2,Y_6\right]&=-Y_6,\,\, &[Y_2,Y_7]&=-Y_7,\,\,&[Y_8,Y_1]&=-Y_1,\,\, &[Y_8,Y_3]&=Y_3,\\ \left[Y_8,Y_4\right]&=0,&\left[Y_8,Y_5\right]&=Y_5,\,\, &[Y_8,Y_6]&=2Y_6,\,\, &[Y_8,Y_7]&=3Y_7,\\
\end{aligned}
$$
we get that $[W,V]\subset V$. Hence, the pair $(W,V)$ becomes a quasi-Lie scheme $S(W,V)$. The theory of quasi-Lie schemes shows that the time-dependent vector field $Y$ can be transformed through any element $g$ of the group of transformations of the scheme $\mathcal{G}(W)$ into a new time-dependent vector field $g_\bigstar Y \in V(\mathbb{R})$, see \cite[Proposition 1]{CGL08}.

The time-dependent transformations associated with elements of $\mathcal{G}(W)$ are
\begin{equation}\label{change}
\left\{\begin{aligned}
x(t)&=\gamma(t)x'(t),\\
v(t)&=\beta(t)v'(t),
\end{aligned}\right.
\end{equation}
with $\gamma(t)>0,\beta(t)>0$ and $\gamma(0)=\beta(0)=1$. The above family of time-dependent changes of variables transforms system (\ref{FOSOR})
into
\begin{equation}
\left\{
\begin{aligned}
\frac{dx'}{dt}&=\frac{\beta(t)}{\gamma(t)}v'-\frac{\dot\gamma(t)}{\gamma(t)}x',\\
\frac{dv'}{dt}&=-\frac{a_0(t)}{\beta(t)}-\gamma(t)\left(b_1(t)v'+\frac{a_1(t)}{\beta(t)}\right)x'-\frac{a_2(t)\gamma^2(t)}{\beta(t)}x'^2\\&\qquad\qquad\qquad\qquad\qquad\qquad-\frac{a_3(t)\gamma^3(t)}{\beta(t)}x'^3-\frac{b_0(t)\beta(t)+\dot \beta(t)}{\beta(t)}v'.
\end{aligned}\right.
\end{equation}
In order to relate this system of first-order differential equations to a SODE, we have to choose $\gamma(t)=K$ for a certain real constant $K$. As $\gamma(0)=1$, then we choose $\gamma(t)=1$ and the previous system becomes
\begin{equation}\label{Sys2}
\left\{
\begin{aligned}
\frac{dx'}{dt}&=\beta(t)v',\\
\frac{dv'}{dt}&=-\frac{a_0(t)}{\beta(t)}-\left(b_1(t)v'+\frac{a_1(t)}{\beta(t)}\right)x'-\frac{a_2(t)}{\beta(t)}x'^2\\ &\qquad\qquad\qquad\qquad\qquad\qquad  -\frac{a_3(t)}{\beta(t)}x'^3-\frac{b_0(t)\beta(t)+\dot \beta(t)}{\beta(t)}v'.
\end{aligned}\right.
\end{equation}
Let us try to relate this system  to one of the Lie
systems of the family (\ref{family}), e.g. the one describing the integral curves of a time-dependent vector field $X_t=f_1(t)X_1+f_2(t)X_2+(f_3(t)/2)(X_3+X_7)+(f_4(t)/4)(X_8-2X_4)$. If we fix $\beta(t)=\sqrt{a_3(t)}$ in (\ref{Sys2}), we obtain
\begin{equation}\label{trasys2}
\left\{
\begin{aligned}
\frac{dx'}{dt}&=\sqrt{a_3(t)}v',\\
\frac{dv'}{dt}&=-\frac{a_0(t)}{\sqrt{a_3(t)}}-\sqrt{a_3(t)}(3v'x'+x'^3)-\frac{a_1(t)}{\sqrt{a_3(t)}}x'-\frac{a_2(t)}{\sqrt{a_3(t)}}(v'+x'^2),
\end{aligned}\right.
\end{equation}
and, consequently, the above system is a Lie system associated with the time-dependent vector field
$$
X_t=\sqrt{a_3(t)}X_1-\frac{a_0(t)}{\sqrt{a_3(t)}}X_2-\frac{a_1(t)}{2\sqrt{a_3(t)}}(X_3+X_7)-\frac{a_2(t)}{4\sqrt{a_3(t)}}(X_8-2X_4).
$$
As the above time-dependent vector field belongs to the family (\ref{family}), the first component of the general solution $(x'(t),v'(t))$ of system (\ref{trasys2}) can be described by means of the expression (\ref{FinSup}) in terms of four generic particular solutions $(x'_{a}(t),v'_{a}(t))$. Moreover, as $\gamma(t)=1$ and $\beta(t)=\sqrt{a_3(t)}$,  then $x(t)=x'(t)$, $x_{a}(t)=x'_{a}(t)$ and $v'_{a}(t)=a_3(t)^{-1/2}dx_{a}/dt(t)$. Using these relations in expression (\ref{FinSup}), we get a time-dependent superposition rule for second-order
Riccati equations. More specifically, the general solution $x=x(t)$ for any instance of second-order Riccati equation of the form (\ref{NLe}) can be written in terms of a set of four particular solutions $x_1=x_1(t)$, $x_2=x_2(t)$, $x_3=x_3(t)$ and $x_4=x_4(t)$ and its derivatives $\dot x_1=\dot x_1(t)$, $\dot x_2=\dot x_2(t)$, $\dot x_3=\dot x_3(t)$ and $\dot x_4=\dot x_4(t)$ as

{\small
\begin{equation}\label{SupRicc}
x=\frac{x_2\widetilde F_{431}-\widetilde G_{2134}\lambda_1-\widetilde G_{3124}\lambda_2-x_3\widetilde F_{412}\lambda_1\lambda_2}{\widetilde F_{431}+(\widetilde F_{124}-\widetilde F_{324})\lambda_1+(\widetilde F_{412}-\widetilde F_{312})\lambda_2+\lambda_1\lambda_2\widetilde F_{421}},
\end{equation}
}where
\begin{multline*}
\widetilde G_{abcd}=x_a(a^{-1/2}_3(t)(\dot x_d-\dot x_c)x_b+a^{-1/2}_3(t)(\dot x_b-\dot x_d)x_c+\\(x_b-x_c)x_bx_c+(x_c-x_b)x_ax_d)+x_d(a^{-1/2}_3(t)(\dot x_c-\dot x_a)x_b+\\a^{-1/2}_3(t)(\dot x_a-\dot x_b)x_c+(x_c-x_b)x_bx_c+(x_b-x_c)x_ax_d),
\end{multline*}
and
\begin{multline*}
\widetilde F_{abc}=a^{-1/2}_3\dot x_a(x_c-x_b)+a^{-1/2}_3\dot x_b(x_a-x_c)+\\+a^{-1/2}_3\dot x_c(x_b-x_a)+(x_a-x_b)(x_b-x_c)(x_c-x_a),
\end{multline*}
with $a,b,c,d=1,\ldots,4$.

Finally, it is just worth noting that, as pointed out in Section \ref{SRSO}, time-dependent superposition rules for a system of SODEs appear naturally related to families of systems admitting the same superposition rule. Indeed, note that  if we had initiated this section by analysing system (\ref{NLe}) with $a_3(t)=1$, we could have proceeded exactly in the same way as before. Nevertheless, when reaching the system (\ref{trasys2}) with $a_3(t)=1$, we could have noticed that all systems with a more general $a_3(t)$ are quasi-Lie systems admitting the same superposition rule as our particular instance. Consequently, this would have shown the existence of a more general system of SODEs, namely, (\ref{NLe}), admitting the same superposition rule.

\section{Appendix}

We have relegated to this appendix various calculations that, although necessary to obtain certain previously stated results, i.e. the common first-integrals for vector fields of $\mathcal{D}$, do not deserve to be detailed in the main body of the article as they do not provide any relevant knowledge.

Recall that a function $F:{\rm T}\mathbb{R}^5\rightarrow \mathbb{R}$ is a common first-integral for every vector field in $\mathcal{D}$ if and only if $\widehat X_1 F=\widehat X_2 F=0$, where $\widehat X_1,\widehat X_2$ are the diagonal prolongations to ${\rm T}\mathbb{R}^5$ of the vector fields $X_1,X_2$. 
Therefore, such a function $F$ must be a solution of the equation 
$$
\widehat X_2F=\sum_{a=0}^4\frac{\partial F}{\partial v_a}=0,
$$
written using the coordinate system $\{x_0,\ldots, x_4,v_0,\ldots,v_4\}$. This equation can be solved using the method of characteristics. Such a method explains that the solutions of the above equation are constant along the solutions, the so-called {\it characteristics curves}, of the system
$$
dx_0=\ldots=dx_4=0,\qquad dv_0=\ldots=dv_4.
$$
So, the characteristics of equation $\widehat X_2F=0$ are curves $(x_0,\ldots,x_4,v_0(s),\ldots,v_4(s))$  such that $\xi_0=v_0(s)-v_4(s)$, $\xi_1=v_1(s)-v_4(s)$, $\xi_2=v_2(s)-v_4(s)$ and $\xi_3=v_3(s)-v_4(s)$ for certain real constants $\xi_0,\ldots,\xi_3$. Thus, there exists a function $F_2:\mathbb{R}^9\rightarrow \mathbb{R}$ such that $F(x_0,\ldots,x_4,v_0,\ldots,v_4)=F_2(x_0,\ldots,x_4,\xi_0,\ldots,\xi_3)$. In other words, function $F$ depends only on the variables $x_0,\ldots,x_4,\xi_0,\ldots,\xi_3$.

Consider the coordinate system $\{x_0,\ldots,x_4,\xi_0,\ldots,\xi_3,v_4\}$ on ${\rm T}\mathbb{R}^5$. In terms of the new coordinate system, the vector field $\widehat X_1$ reads
\begin{multline*}
\widehat X_1=\sum_{a=0}^3\left(\xi_a\frac{\partial}{\partial x_a}-\left(3x_a \xi_a-x_4^3+x_a^3\right)\frac{\partial}{\partial \xi_a}\right)-(x_4^3+3x_4v_4)\frac{\partial}{\partial v_4}+\\+v_4\left(\sum_{a=0}^4\frac{\partial}{\partial x_a}-3\sum_{a=0}^3(x_a-x_4)\frac{\partial}{\partial \xi_a}\right).
\end{multline*}
As we assumed that $F$ is a first-integral of the distribution $\mathcal{D}$, it is a first-integral for the vector fields $\widehat X_1$ and $\widehat X_2$. Taking into account that $F$ is a solution of the equation $\widehat X_2F=0$ and, therefore, it depends only on the variables $x_0,\ldots,x_4,\xi_0,\ldots,\xi_3$, the equation $\widehat X_1F=0$ yields
$$
\sum_{a=0}^3\left(\xi_a\frac{\partial F_2}{\partial x_a}-\left(3x_a \xi_a-x_4^3+x_a^3\right)\frac{\partial F_2}{\partial \xi_a}\right)+v_4\left(\sum_{a=0}^4\frac{\partial F_2}{\partial x_a}-3\sum_{a=0}^3(x_a-x_4)\frac{\partial F_2}{\partial \xi_a}\right)=0,$$
and in view of the dependence of the function $F_2$, we get that there exist two vector fields 
\begin{equation*}
Z_1=\sum_{a=0}^3\left(\xi_a\frac{\partial}{\partial x_a}-\left(3x_a \xi_a-x_4^3+x_a^3\right)\frac{\partial }{\partial \xi_a}\right),\,
Z_2=\sum_{a=0}^4\frac{\partial}{\partial x_a}-3\sum_{a=0}^3(x_a-x_4)\frac{\partial}{\partial \xi_a},
\end{equation*}
such that $\widehat X_1F=Z_1F_2+v_4Z_2F_2=0$. Therefore, we have that $Z_1F_2=0$ and $Z_2F_2=0$. In consequence, the function $F$ is a first-integral of the distribution $\mathcal{D}$ if and only if it is a first-integral for the vector fields $Z_1$ and $Z_2$ depending on the variables $x_0,\ldots,x_4,\xi_0,\ldots,\xi_3$. 

The first-integrals of the vector field $Z_2$ depending just on the above variables can be determined by means of the characteristic curves given by the following system
\begin{equation*}
 d(x_a-x_4)=0, \qquad dx_4=-\frac{d\xi_a}{3(x_a-x_4)},\qquad a=0,\ldots,3.
\end{equation*}
Hence, such first-integrals depend on the functions 
$$
\left\{
\begin{aligned}
\eta_a&=x_a-x_4,\\
\phi_a&=3\eta_ax_4+\xi_a,\\
\end{aligned}\right.\qquad a=0,\ldots,3,
$$
Consequently, given a first-integral $F$ of the distribution $\mathcal{D}$, there exists a function $F_3:\mathbb{R}^8\rightarrow\mathbb{R}$ such that $F(x_0,\ldots,x_4,v_0,\ldots,v_4)=F_3(\eta_0,\ldots,\eta_3,\phi_0,\ldots,\phi_3)$, i.e. the function $F$  actually only depends on the variables $\eta_0,\ldots,\eta_3,\phi_0,\ldots,\phi_3$.

Taking into account the dependence of the function $F$ in the above variables, the equation $Z_1F=0$ reads in the coordinate system $\{\eta_0,\ldots,\eta_3,\phi_0,\ldots,\phi_3,x_4,v_4\}$,
\begin{equation*}
\sum_{a=0}^3\left[(\phi_a-3\eta_a x_4)\frac{\partial F}{\partial \eta_a}-\left(3\eta_a\phi_a-6\eta_a^2x_4+3\eta_a x_4^2+\eta_a^3\right)\frac{\partial F}{\partial \phi_a}\right]=0.
\end{equation*}
Collecting terms with different powers of $x_4$, we obtain that $Z_2F=\Omega_0F_3+x_4\Omega_1F_3-3x_4^2\Omega_2F_3=0$, with
\begin{equation*}
\begin{gathered}
\Omega_0=\sum_{a=0}^3\left(\phi_a\frac{\partial }{\partial \eta_a}-\left(3\eta_a\phi_a+\eta_a^3\right)\frac{\partial}{\partial \phi_a}\right),\qquad 
\Omega_1=\sum_{a=0}^3\left(-3\eta_a\frac{\partial}{\partial \eta_a}+6\eta_a^2\frac{\partial}{\partial \phi_a}\right),\\
\Omega_2=\sum_{a=0}^3\eta_a\frac{\partial}{\partial \phi_a}.
\end{gathered}
\end{equation*}
As $F_3$ does not depend on $x_4$ in the  set of coordinates we have chosen, then $\Omega_aF_3=\Omega_aF=0$, for $a=0,1,2$. The method of characteristics for the equation $\Omega_1F_3=0$ implies that there exists a function $F_4:\mathbb{R}^7\rightarrow\mathbb{R}$ such that $F(x_0,\ldots,x_4,v_0,\ldots,v_4)=F_4(\delta_0,\ldots,\delta_3,L_1,L_2,L_3)$ with
\begin{equation*}
\left\{
\begin{aligned}
\delta_a&=\phi_a+\eta_a^2,\qquad &a&=0,\ldots,3,\\
L_a&=\frac{\eta_a}{\eta_0},\qquad &a&=1,2,3.\\
\end{aligned}\right.
\end{equation*}
Now, using the coordinate system $\{\delta_0,\ldots,\delta_3,L_1,L_2,L_3,\eta_0,x_4,v_4\}$, we get 
$$
\Omega_2F=\Omega_2F_4=\eta_0\left(\frac{\partial F_4}{\partial \delta_0}+\sum_{a=1}^3 L_a\frac{\partial F_4}{\partial \delta_a}\right)=0,
$$
and, repeating the previous procedure, we see that there exists a function $F_5:\mathbb{R}^6\rightarrow\mathbb{R}$ such that $F(x_0,\ldots,x_4,v_0,\ldots,v_4)=F_5(L_1,L_2,L_3,\Delta_1,\Delta_2,\Delta_3)$, where $\Delta_a=L_a\delta_0-\delta_a$ and  $a=1,2,3$.

As we have shown that finding a first-integral $F$ for the vector fields of the distribution $\mathcal{D}$ reduces to looking for a first-integral $F_5$  of the vector field $\Omega_0$ depending on the variables $L_1,L_2,L_3, \Delta_1,\Delta_2,\Delta_3$, we still have to analyse the condition $\Omega_0F_5=0$ to determine completely the form of the first-integrals for the distribution $\mathcal{D}$. 

By choosing the coordinate system $\{L_1,L_2,L_3,\Delta_1,\Delta_2,\Delta_3,\delta_0,\eta_0,x_4,v_4\}$, the equation $\Omega_0F_5=0$ reads
$$
\eta_0^2\sum_{a=1}^3\left((L_a-L_a^2)\frac{\partial F_5}{\partial L_a}-L_a\Delta_a\frac{\partial F_5}{\partial \Delta_a}\right)-\left(\sum_{a=1}^3\Delta_a\frac{\partial F_5}{\partial L_a}+\delta_0\Delta_a\frac{\partial F_5}{\partial \Delta_a}\right)=0,
$$
and considering the vector fields
$$
\Xi_1=\sum_{a=1}^3(L_a-L_a^2)\frac{\partial }{\partial L_a}-L_a\Delta_a\frac{\partial}{\partial \Delta_a},\quad \Xi_2=\sum_{a=1}^3\Delta_a\frac{\partial }{\partial L_a},\quad 
\Xi_3=\sum_{a=1}^3\Delta_a\frac{\partial }{\partial \Delta_a},
$$
and the form of the function $F_5$, the equation $\Omega_0F_5=0$ implies that $\Xi_1F_5=\Xi_2F_5=\Xi_3F_5=0$. 

If we apply the method of characteristics to the equation $\Xi_3F_5=0$, it yields that there exists a function $F_6:\mathbb{R}^5\rightarrow\mathbb{R}$ such that $F(x_0,\ldots,x_4,v_0,\ldots,v_4)=F_6(L_1,L_2,L_3,\pi_2,\pi_3)$, with $\pi_2=\Delta_2\Delta_1^{-1}$ and $\pi_3=\Delta_3\Delta_1^{-1}$.
Moreover, as $F_5$ also satisfies the equation $\Xi_2F_5=\Xi_2F_6=0$, we obtain that $F_6$ only  depends on the variables $\pi_2$, $\pi_3$, $\Gamma_2=\pi_2 L_1-L_2$ and $\Gamma_3=\pi_3L_1-L_3$, i.e. there exists a function $F_7:\mathbb{R}^4\rightarrow \mathbb{R}$ such that $F(x_0,\ldots,x_4,v_0,\ldots,v_4)=F_6(L_1,L_2,L_3,\pi_2,\pi_3)=F_7(\pi_2,\pi_3,\Gamma_2,\Gamma_3)$.

Finally, the conditions $\Xi_1F=\Xi_1F_7=0$ imply that $\Xi_1F_7=\Upsilon_2F_7+L_1\Upsilon_1F_7=0$, where
$$
\Upsilon_1F_7=(\pi_2-\pi_2^2)\frac{\partial F_7}{\partial \pi_2}+(\pi_3-\pi_3^2)\frac{\partial F_7}{\partial \pi_3}-\pi_2\Gamma_2\frac{\partial F_7}{\partial \Gamma_2}-\pi_3\Gamma_3\frac{\partial F_7}{\partial \Gamma_3}=0$$
and
$$
\Upsilon_2F_7=(\Gamma_2+\Gamma_2^2)\frac{\partial F_7}{\partial \Gamma_2}+(\Gamma_3+\Gamma_3^2)\frac{\partial F_7}{\partial \Gamma_3}+\Gamma_2\pi_2\frac{\partial F_7}{\partial \pi_2}+\Gamma_3\pi_3\frac{\partial  F_7}{\partial \pi_3}=0.
$$
The first equation implies that there exists a function $F_8:\mathbb{R}^3\rightarrow \mathbb{R}$ which satisfies that $F(x_0,\ldots,x_4,v_0,\ldots,v_4)=F_7(\pi_2,\pi_3,\Gamma_2,\Gamma_3)=F_8(\Psi_0,\Psi_1,\Psi_2),$ where $\Psi_0=(\pi_2-1){\Gamma_2}^{-1}$, $\Psi_1=(\pi_3-1){\Gamma_3}^{-1}$ and $\Psi_2={\pi_2}{\pi_3}^{-1}(1-\pi_3)(1-\pi_2)^{-1}$.
So, the equation $\Upsilon_2F_7=\Upsilon_2F_8=0$ can be cast into the form
$$
(1-\Psi_0)\frac{\partial F_7}{\partial \Psi_0}+(1-\Psi_1)\frac{\partial F_7}{\partial \Psi_1}-\Psi_2\frac{\Psi_1-\Psi_0}{\Psi_0\Psi_1}\frac{\partial F_7}{\partial \Psi_2}=0,$$
and we obtain that $F$ is an arbitrary function of the functions $\Lambda_1=\frac{1-\Psi_0}{1-\Psi_1}$ and $\Lambda_2=\frac{\Psi_0\Psi_2}{\Psi_1}$. In particular, undoing all the changes of variables performed along this Section, we get that
$$
\Lambda_1(p)=\frac{F_{431}(p)F_{210}(p)}{F_{421}(p)F_{310}(p)}\qquad {\rm and}\qquad \Lambda_2(p)=\frac{F_{431}(p)F_{420}(p)}{F_{421}(p)F_{430}(p)}
$$
are two independent first-integrals for the vector fields of the distribution $\mathcal{D}$. 
\section{Conclusions and outlook}
We have defined and analysed the concepts of superposition rule, time-dependent superposition rule, and free superposition rule for systems of SODEs. Several results concerning the existence of such superposition rules have been proved. Posteriorly, our theoretical achievements have been illustrated by means of the study of a number of SODEs appearing in the Physics and Mathematics literature. Several new SODE Lie systems have been described and a common superposition rule for all of them has been derived. Such a superposition rule has posteriorly been used, with the aid of the theory of quasi-Lie systems, to get a time-dependent superposition rule for second-order Riccati equations.

In the future, we expect to continue the analysis the properties of superposition rules for systems of SODEs as well as to investigate the generalization of the techniques depicted throughout this work to systems of higher-order differential equations. As an application, we hope to apply the results obtained in order to study new systems of second- and higher-order differential equations. In particular, it is specially interesting analysing the higher members of the Riccati hierarchy in other to develop new methods to determine solutions for those PDEs whose B\"acklund transformations are described by members of such a  hierarchy.

\section*{Acknowledgments} 

Partial financial support by research projects E24/1 (DGA), MTM2009-08166-E and MTM\\2009-11154 is acknowledged. 


\end{document}